\documentclass[journal,a4paper]{IEEEtran}

\usepackage{ifpdf}

\usepackage{cite}

\ifCLASSINFOpdf
  \usepackage[pdftex]{graphicx}
\else
\fi
\usepackage{epstopdf}
\usepackage{epsfig}

\usepackage[cmex10]{amsmath}
\usepackage{amsthm}
\usepackage{array}

\usepackage[font=footnotesize]{subfig}
\usepackage{url}

\newtheorem{defn}{Definition}
\newtheorem{obs}{Observation}
\newtheorem{thm}{Theorem}
\newtheorem{lem}{Lemma}
\begin{document}

\title{Quality Sensitive Price Competition in Spectrum Oligopoly}

\author{Arnob~Ghosh and
        Saswati~Sarkar
\thanks{The authors are with the Department
of Electrical and Systems Engineering, University Of Pennsylvania, Philadelphia,
PA, USA.Their  E-mail ids are arnob@seas.upenn.edu and swati@seas.upenn.edu}}



\maketitle


\begin{abstract}
We investigate a spectrum oligopoly  where primary users allow secondary access  in lieu of financial remuneration. Transmission qualities of the licensed bands
fluctuate randomly. Each primary needs to select the price of its channel with the knowledge of its own channel state but not that of its competitors. Secondaries choose among the channels available on sale  based on  their states   and prices. We formulate the price selection as a non-cooperative game and prove that a symmetric Nash equilibrium (NE) strategy profile exists uniquely. We explicitly compute this strategy profile and analytically and numerically evaluate its efficiency.  Our structural results provide certain key insights about the unique symmetric NE.
\end{abstract}

\section{Introduction}
Recent investigations augur that demand for mobile broadband -- driven by the large scale proliferation of wireless industry  -- will surpass the availability of wireless spectrum in imminent future. Yet, as recent measurements suggest, the licensed bands remain largely under-utilized.
A reasonable conjecture therefore is that unlicensed access of  idle (but licensed) spectrum bands,
 commonly referred to as secondary spectrum access, would avert the impending crisis. Recently, FCC has  legalized the access of TV white space spectrum, and
 the advent of cognitive radios together with the design of a plethora of sophisticated algorithms  have enabled intelligent   selection of  bands.
Large-scale secondary spectrum access can  not however be realized
   only through the availability of the enabling technology and the regulatory progress:  secondary access must also be rendered  profitable for the license holders.  Accordingly,
we investigate a spectrum oligopoly \cite{mwg}  where license holders (hitherto referred to as primaries)  allow unlicensed users (hitherto referred to as secondaries), in lieu of financial remuneration,  access to the channels  (licensed bands) that are not in use. Different channels offer different transmission rates to the secondaries
depending on their states which evolve randomly and reflect the usage levels of the primaries as also transmission quality fluctuations owing to fading. 
Each primary quotes a price for the channel that it offers and secondaries select among
the available channels depending on the states and the prices quoted. Thus,
if a primary quotes a high price, it will earn a large profit if it sells its channel, but may not be able to sell at all; on the other hand a low price will enhance the probability of a sale but also fetch lower profits in the event of a sale.

Price selection in oligopolies may be modeled as a non-cooperative game and has naturally been extensively investigated in economics, implementing \textit{Bertrand Game} \cite{mwg} and its modifications \cite{Osborne, Kreps, Chawla}.
However, all the above papers ignore the uncertainty of competition which distinguishes spectrum markets from standard oligopolies: a primary knows the state of its channel but does not know those of  its competitors before deciding the price
for its channel. Pricing in communication services have been explored to a great extent (\cite{courcoubetis} prsenets a brief overview). References \cite{Ileri, Mailespectrumsharing, Mailepricecompslotted, Xing, Niyatospeccrn, Niyatomultipleseller} have analyzed price competition among spectrum providers. References \cite{Niyatospeccrn, Niyatomultipleseller} modeled price competition among multiple players. But all the above papers suffer from drawbacks: first, they did not assay uncertainty of states of channels of competitors ; second, most of them did not explicitly determine a Nash Equilibrium (NE) (exceptions are \cite{Mailespectrumsharing},\cite{Niyatospeccrn}).   On the other hand, the  papers that consider uncertainty of competition, namely \cite{Gaurav1,Gaurav2, gauravjsac, Janssen,Kimmel},  assume that the commodity on sale can be in one of two states: available or otherwise. This assumption does not capture different transmission qualities offered by
 the available channels. The consideration of the latter significantly complicates the  analysis of the game. A primary may now need to employ different pricing strategies for different states, while in the former case a single pricing strategy will suffice as a price need not be quoted for an unavailable commodity.
Our investigation seeks to contribute in this space.



We have modeled the price selection as a game with primaries as the players (Section~\ref{sec:model}) and seek an NE  pricing strategy.  We consider that the preference of the secondaries can be captured by a penalty function which associates a penalty value to each channel that is available for sale depending on its state and price quoted. Given the state of a channel, there is a one-to-one correspondence between the price quoted and the penalty perceived by a secondary. Thus, the strategy for selection of a price for a channel in a given state may be equivalently represented as a strategy for selection of penalty that the channel offers to a secondary. Since prices and therefore the penalties take real values, the strategy set of the players are continuous; also the payoff functions for the primaries turn out to be discontinuous. Thus, classical results
do not guarantee the existence, let alone the uniqueness,  of an NE. In addition, existing literature does not provide algorithms for computing an NE unlike when the strategy set is finite\cite{myerson}. Starting from a general set of strategy profiles for the primaries which allows for selecting penalties using arbitrary probability distributions, we show that for a large class of penalty functions, there exists a unique symmetric NE strategy profile, which we explicitly compute (Section~\ref{sec:results}). Our analysis reveals several interesting insights about the structure of the symmetric NE. First, we learn that  if a channel in state $i$ provides a higher transmission rate to a secondary than that in state $j$, then the symmetric NE strategy profile selects the penalties for $i, j$ respectively from ranges $[L_i, U_i], [L_j, U_j]$ where $U_i \leq L_j.$ Thus, a secondary will always prefer a channel in state $i$ to a channel in state $j$ considering both the prices and the states.  This negates the intuition that prices ought to be selected for the states so as to render them equally preferable to a secondary - symmetric NE strategy profiles in fact price the channels so as to retain the preference order provided by the states. The analysis also reveals that
the unique symmetric NE strategy profile consists ``nice'' probability distributions in that they are  continuous and strictly increasing;  the former  rules out pure strategy symmetric NEs and the latter ensures that the support sets are contiguous. Finally, utilizing the explicit computation algorithm for the symmetric NE strategies, we analytically and numerically investigate the reduction in expected profit suffered under the unique symmetric NE pricing strategies as compared to  the maximum possible value allowing for collusion among primaries (Section~\ref{numerical}).  

\textit{All the proofs are deferred to the Appendix}.


\section{System Model}
\label{sec:model}

  We consider a spectrum market with $n$ primaries and $m$ secondaries. We will initially consider the case that the primaries know $m$, later generalize our results for random, apriori unknown $m$. Each primary has access to a channel which can be in states $0, 1, \ldots, n$, where state $i$ provides a lower transmission rate to a secondary than state $j$ if $i < j$ and state $0$ arises when
  the channel is not available for sale and provides $0$ transmission rate.
  Different channels constitute disjoint frequency bands leased by the primaries.
  A channel is in state $i \geq 1 $ w.p. $q_i$ and in state $0$ w.p. $1-q$ where $q = \sum_{i=1}^n q_i$, independent of the states of other channels.
    If a primary quotes a price $p$ for a channel in state $i$, then the channel offers a penalty $g_i(p)$ to a secondary. Each $g_i(\cdot)$ is continuous, strictly increasing in its argument, and therefore invertible. We denote $f_i(\cdot)$ as the inverse of $g_i(\cdot)$; clearly $f_i(\cdot)$ is continuous and strictly increasing in its argument as well.   No secondary buys a channel whose penalty is higher than $v$, and as the name suggests a secondary prefers a channel with a lower penalty (a secondary's preference depends entirely on the penalty). Thus, we must have $g_i(p)>g_j(p)$ and $f_i(x)<f_j(x)$ for each $x, p$ and $i<j$.   Each primary also incurs a transition cost $c>0$ for an available channel, and therefore never selects a price lower than $c.$  We  assume that
\begin{align}\label{con1}
\dfrac{f_j(y)-c}{f_k(y)-c}<\dfrac{f_j(x)-c}{f_k(x)-c}  \ \mbox{ for all } x>y > g_j(c), j<k
\end{align}
A large  class of penalty functions $g_i(\cdot)$  satisfy the above property required of the corresponding inverses, e.g., $g_i(p) = \zeta\left(p - h(i)\right), g_i(p) = \zeta\left(p/h(i)\right)$ where
 $\zeta(\cdot)$ is continuous, strictly increasing function and $h(\cdot)$ is strictly increasing, $g_i(p) = p^r - h(i), g_i(p) = p^rh(i), g_i(p)=\exp(p) - h(i), g_i(p)=\log(p)-h(i)$ for $r > 0$ and a strictly increasing $h(\cdot).$ In addition,
 $g_i(\cdot)$ such that the inverses are of the form $f_i(x)=h_1(x)+h(i), f_i(x) = h_1(x)h(i)$,  where $h_1(\cdot)$ is continuous and strictly increasing; $h(\cdot)$ is sirictly increasing, satisfy the above assumption.

If primary $i$ quotes   a price $p$ for its channel then its  profit(payoff)  is \\
\begin{equation*}
\begin{cases} p-c  & \text{if the primary sells its channel}\\
0 & \text{otherwise}
\end{cases}
\end{equation*}
Note that if $Y$ is the number of channels offered for sale for which the penalties are upper bounded by $v$, then
those with $\min(Y, m)$ lowest penalties are sold  since secondaries select channels in increasing order of penalties. The ties among channels with identical penalties are broken randomly and symmetrically
   among the primaries. Also, note that utilities of primaries are not continuous functions of their actions.


 Each primary selects the penalty for its channel with the knowledge of the state of the channel, but without knowing the states of the other channels; a primary however knows $l, m, n, q_1, \ldots, q_n$. Note that the choice of the penalty uniquely determines the price since there is a one-to-one correspondence between the two given the state of a channel.
Primary $i$ chooses its penalty using an arbitrary probability distribution function (d.f.) \footnote{Here probability distribution refers cumulative distribution function. Recall the definition of cumulative distribution function (d.f) of a random variable $X$ is the function
$G(x)=P(X\leq x) \quad x\in \Re$ \cite{df}}$\psi_{i,j}(.)$
when its channel is in state $j \geq 1$. If $j = 0$ (i.e., the channel is unavailable),   $i$ chooses a penalty of $v+1$: this is equivalent
 to considering that such a channel is not offered for sale as no secondary
   buys a channel whose penalty exceeds $v$. For $j \geq 1$,  each primary selects its price   so as to maximize  its expected profit. Thus, if $m\geq l$,  primaries  select the highest penalty for each state $1,\ldots, n$, since   all available channels will be sold. So, we consider $m<l$. $S_i=(\psi_{i,1},....,\psi_{i,n})$ denotes the strategy of primary $i$, and $(S_1,...,S_l)$ denotes the strategy profile of all primaries (players).
\begin{defn}
$S_{-i}$ denotes the strategy profile of primaries other than $i.$
$E\{u_{i, j}(\psi_{i,j},S_{-i})\} $ denotes the expected profit when primary $i$'s channel is in state $j$ and it  uses strategy $ \psi_{i, j}(\cdot) $  and other primaries use strategy $S_{-i}$.
\end{defn}
\begin{defn}
A \emph{Nash  equilibrium}  $(S_1, \ldots, S_n)$ is a strategy profile such that no primary can improve its expected profit by unilaterally deviating from its strategy  \cite{mwg}.  So, with $S_i=(\psi_{i,1},....,\psi_{i,n})$, $(S_1, \ldots, S_n)$, is  a  Nash equilibrium (NE) if for each primary $i$ and channel state $j$
\begin{align}
E\{u_{i,j}(\psi_{i,j},S_{-i})\}\geq E\{u_{i,j}(\tilde{\psi}_{i,j},S_{-i})\} \ \forall \ \tilde{\psi}_{i, j}.
\end{align}
An NE $(S_1, \ldots, S_n)$  is a \emph{symmetric NE} if $S_i = S_j$ for all $i, j.$
\end{defn}

The above game is a symmetric one since  primaries have the same action sets, payoff functions and their channels are statistically identical. We therefore consider only symmetric NEs\footnote{For a symmetric game, an asymmetric NE is rarely realized. For example. for two players, if  $(S_1, S_2)$ is an NE,
 $(S_2, S_1)$ is also an NE. The realization of such an NE is possible only when each player knows whether
 the other uses $S_1$ or $S_2$.  This complication  is somewhat alleviated for a symmetric NE as all players play the same strategy; this complication is eliminated only when there is a unique symmetric NE. Note that, there are plethora of examples of symmetric games \cite{mwg}, which have multiple NEs. We prove that that there is a unique symmetric NE for the game we consider. }.  Clearly, for any symmetric NE, we can represent the strategy of any primary as $S = (\psi_{1}(.), \psi_2(.),.....,\psi_n(.)) $ where we drop the index corresponding to the primary.

  Let $\phi_j(x)$ denote the expected profit of a primary whose channel is in state $j$ and who selects
  a penalty $x$ and $r(x)$ denote the probability that a channel quoted at penalty $x$ is sold.  Note that the dependence of $\phi_j(x), r(x)$ on the strategy profile of the primaries is
  not explicitly indicated to ensure notational simplicity. Also, note that $r(x)$ does not depend on the state of the channel since secondaries select the channels based only on the penalties. Next,
\begin{align}\label{ex1}
\phi_j(x)=(f_j(x)-c)r(x).
\end{align}
(recall that the inverse of the penalty function $g_j(\cdot)$,  $f_j(\cdot)$, provides the price
 that corresponds to penalty $x$ and channel state $j$).
\begin{defn}\label{br}\label{du}
 A \emph{best response} penalty for  a channel in state  $j\geq 1$  is $x$ if and only if
\begin{equation*}
\phi_j(x)=\underset{y\in \Re}{\text{sup}}\phi_j(y).
\end{equation*}
Let $u_{j,max} = \phi_j(x)$ for a best response $x$ for state $j$, $j \geq 1$ i.e.,  $u_{j,max}$ is the maximum expected profit  that a primary earns under NE strategy profile, when its channel is in state  $j$, $j \geq 1$ .
\end{defn}
\section{A symmetric NE: Existence, Uniqueness and Computation}
\label{sec:results}
 First, we identify key structural properties of a symmetric  NE (should it exist).
  Next we show that the above properties leads to a unique strategy profile which we explicitly compute  - thus the symmetric NE is unique should it exist.  We finally prove that the strategy profile resulting from the structural properties above is   indeed a symmetric NE thereby establishing existence.

\subsection{Structure of a  symmetric NE}\label{sec:structure}

We start with by providing some important properties that any symmetric NE $\left(\psi_1(\cdot), \ldots, \psi_n(\cdot)\right)$ must satisfy.
\begin{thm}
\label{thm1}
$\psi_{i}(.), i\in \{1,..,n\}$ is a continuous probability distribution.
\end{thm}
The above theorem rules out any pure strategy symmetric NE.

\begin{defn}\label{dlu}
 We denote the lower and upper endpoints of the support set\footnote{The support set of a probability distribution is the smallest closed set such that the probability of its complement is $0.$}   of $\psi_i(.)$ as $L_i$ and $U_i$ respectively i.e.
\begin{equation*}\label{n77}
L_i=\inf\{x: \psi_{i}(x)>0\}
\end{equation*}
\begin{equation*}\label{n77a}
U_i=\inf\{x: \psi_i(x)=1\}
\end{equation*}
\end{defn}
We next show that the support sets are ordered in increasing order of the state indices.
\begin{thm}
\label{thm2}
$U_i\leq L_j$, if $j<i$
\end{thm}
We finally rule out any ``gaps'' inside the support sets and between the support sets for different
$\psi_i(\cdot)$,  $i=1,..,n$. This also establishes that $\psi_i(\cdot)$ is strictly increasing in $[L_i, U_i].$
\begin{thm}
\label{thm3}
The support set of $\psi_i(.), i=1,..,n$ is $[L_i, U_i]$  and $U_i=L_{i-1}$ for $i=2,..,n$, $U_1=v$.
\end{thm}
\textit{Remark}: The structure of the symmetric NE identified in Theorems 1 to 3 provide several interesting insights:
\begin{itemize}
\item Theorem~\ref{thm2} implies that the primaries select the highest penalties for the worst states.
 The primaries therefore do not strive to render all states equally preferable to the secondaries through price selection.
 \item Theorems~\ref{thm1} and \ref{thm3} reveal that the symmetric NE strategy profile consists of
  ``well-behaved'' distribution functions.
 \end{itemize}

\subsection{Computation and Uniqueness of a Symmetric NE}
\label{sec:computation}
We now show that the structural properties of a symmetric NE identified in
Theorems~\ref{thm1}, \ref{thm2}, \ref{thm3} are satisfied by a unique strategy profile,
 which we explicitly compute. This proves the uniqueness of a symmetric NE subject to existence.
 We start with the following definitions.
 \begin{eqnarray}
w(x)&=&\sum_{i=m}^{l-1}\dbinom{l-1}{i}x^i(1-x)^{l-i-1}\label{d4}\\
w_i&=&w(\sum_{j=i}^{n}q_j) \quad \text{for} i=1,...,n \quad \text{and} w_{n+1}=0\label{n50}
\end{eqnarray}
Clearly, for $x \in [0, 1]$,  $w(x)$ is the probability of at least m successes out of l-1 independent Bernoulli trials, each of which occurs with probability $x.$ Note that $w(\cdot)$ is continuous and strictly increasing in $[0,1]$ \cite{walsh}, so its inverse exists.
Note that $w_i>w_j$ if $i<j, i,j\in\{1,...,n\}$ as $w_i$ is the success probability of at least m successes out of (l-1) independent Bernoulli Events, where each of which occurs with probability $\sum_{j=i}^{n}q_j$.

\begin{lem} \label{lu} For $1 \leq i \leq n$,
\begin{eqnarray}u_{i,max} & = & p_i-c \nonumber\\
\mbox{where, } p_i& =& c+(f_i(L_{i-1})-c)(1-w_i)
\label{n51}\\
\mbox{and } L_i&=&g_i(\dfrac{p_i-c}{1-w_{i+1}}+c), L_{0}=v\label{n52}
\end{eqnarray}
\end{lem}
Using \eqref{n51} and \eqref{n52}, $u_{i,max}, L_i$ can be computed recursively starting from $i=1.$ Note that as $1-w_i>0,\forall i\in\{1,\ldots,n\}$, thus, $p_i-c>0$. Hence, from the definition of $L_k$ (\ref{n52}), it is evident that 
\begin{align}\label{imp}
f_k(L_k)>c
\end{align}
Expressions of $L_i$ and $p_i$ are used in the following lemma to determine the unique $\psi_i(\cdot)$, if it exists
\begin{lem} \label{lm:computation}
A symmetric NE strategy profile  $\left(\psi_1(\cdot), \ldots, \psi_n(\cdot)\right)$
comprises of: \begin{align}\label{c5}
\psi_i(x)= & 
 0 ,  \text{if} \ x<L_i\nonumber\\
& \dfrac{1}{q_i}(w^{-1}(\dfrac{f_i(x)-p_i}{f_i(x)-c})-\sum_{j=i+1}^{n}q_j), \text{if} \ L_{i-1}\geq x\geq L_i\nonumber\\
& 1,  \text{if} \ x>L_{i-1}
\end{align}
where $L_i, i=1,..,n$ are as defined in (\ref{n52}) and $L_0=v$.
\end{lem}

Next lemma will ensure that $\psi_i(\cdot)$ as defined in lemma~\ref{lm:computation} is indeed a d.f.

\begin{lem}\label{dcont}
$\psi_i(\cdot)$ as defined in Lemma~\ref{lm:computation} is a strictly increasing and continuous distribution function.
\end{lem}


\subsection{Existence of a symmetric NE}
\label{sec:existence}
In this  section, We prove that symmetric strategy profile identified in previous section is indeed a NE strategy profile.
\begin{thm}\label{thm4} $\left(\psi_1(\cdot), \ldots, \psi_n(\cdot)\right) j=1,..,n$ as defined in lemma~\ref{lm:computation} is  a symmetric NE.
\end{thm}


\paragraph{Remark} Note that all our results readily generalize to allow for random number of secondaries (M) with probability mass functions (p.m.f.) $\Pr(M=m)=\gamma_m$. A primary does not have the exact realization of number of secondaries, but it knows the p.m.f. . We only have to redefine  $w(x)$ as-
\begin{eqnarray}
\sum_{k=0}^{\max(M)}\gamma_k\sum_{i=k}^{l-1}\dbinom{l-1}{i}x^i(1-x)^{l-1-i}
\end{eqnarray}
and $w_{n+1}=\gamma_0$ .\\

\section{Performance evaluation of the symmetric NE }
\label{numerical}
\begin{defn}
Let $R_{NE}$ denote  the total expected profit at Nash equilibrium. Then,
\begin{eqnarray}
R_{NE}=l.\sum_{i=1}^{n}(q_i.(p_i-c))
\end{eqnarray}
\end{defn}

\begin{lem}\label{eff}
Let $c_{j}=g_j(c), j=1,..,n$.
\begin{enumerate}
\item If $m\geq (l-1)(\sum_{j=1}^n q_j +\epsilon)$ for some $\epsilon>0$, then $R_{NE}\rightarrow l.\sum_{j=1}^{n}q_j.(f_{j}(v)-c)$ as $l\rightarrow \infty.$
\item If $ (l-1)(\sum_{j=i-1}^{n}q_j-\epsilon)\geq m \geq (l-1)(\sum_{j=i}^{n}q_j+\epsilon), i\in\{2,..,n\}$, for some $\epsilon>0$, then
$R_{NE}\rightarrow l.\sum_{j=i}^{n}q_j.(f_{j}(c_{i-1})-c)$ as $l\rightarrow \infty.$
\item If $m\leq (l-1)(q_n-\epsilon)$ for some $\epsilon>0$, then $R_{NE}\rightarrow 0$ as $l\rightarrow \infty.$
\end{enumerate}
\end{lem}
Note that, for $j>i$, $c_j<c_i$ (as $g_j(c)<g_i(c)$), thus, asymptotically $R_{NE}$ decreases as $m$ decreases. This is expected as  competition increases with decrease in $m$, and thus prices
 are chosen progressively closer to the lower limit, that of the transition cost,  $c.$

\begin{defn}
Let $R_{OPT}$ be  the maximum  expected profit earned through collusive selection of
 prices by the primaries. \emph{Efficiency} $\eta$ is defined as $\dfrac{R_{NE}}{R_{OPT}}.$
\end{defn}
Efficiency is a measure of the reduction in the expected profit owing to competition.
Asymptotic behavior of $\eta$ is characterized by the following lemma:
\begin{lem}\label{thresh}
\begin{enumerate}
 \item If $m\geq (l-1)(\sum_{j=1}^{n}q_j+\epsilon)$ for some $\epsilon>0$, then $\eta\rightarrow 1$ as $l\rightarrow \infty$.
     \item If $m\leq (l-1)(q_n-\epsilon)$ for some $\epsilon>0$, then $\eta\rightarrow 0$ as $l\rightarrow \infty$.
         \end{enumerate}
\end{lem}

The lemma does not characterize the asymptotic limits for $\eta$  for $m \in [(l-1)\sum_{j=1}^{n}q_j,   (l-1)q_n].$
However, our numerical computation reveals that $\eta$ increases from $0$ to $1$ with increase in $m$ (figure ~\ref{fig:Efficiency}) - the variation is largely monotonic barring for a few discrepancies owing
to $n, m$ being finite.  Intuitively, demand increases with increase in  $m$; thus primaries set their penalties close to the highest possible value for all states which leads to higher efficiency. On the other hand, when $m$ decrease, competition becomes intense and primaries chooses prices close to $c$ and expected profits under the symmetric NE decreases: $R_{NE}$ is very small as lemma ~\ref{eff} reveals. But, if primaries collude,  primaries can judiciously offer only the channels of highest possible states to the secondaries to gain a large profit. Hence, the decrease in $R_{OPT}$ with decrease in $m$ is slower,  which leads to lower efficiency for low $m.$

Similarly, when $q_i$s increases, competition becomes intense and primaries chooses price closer to $c$, hence $R_{NE}$ decrease. But, when primaries collude, they still sell at highest possible penalty for a channel and hence $\eta$ decrease. On the other hand, when, $q_i$s decrease, primaries set their prices closer to highest possible values for all states and thus, $\eta$ increase.

\begin{figure}
\includegraphics[width=90mm]{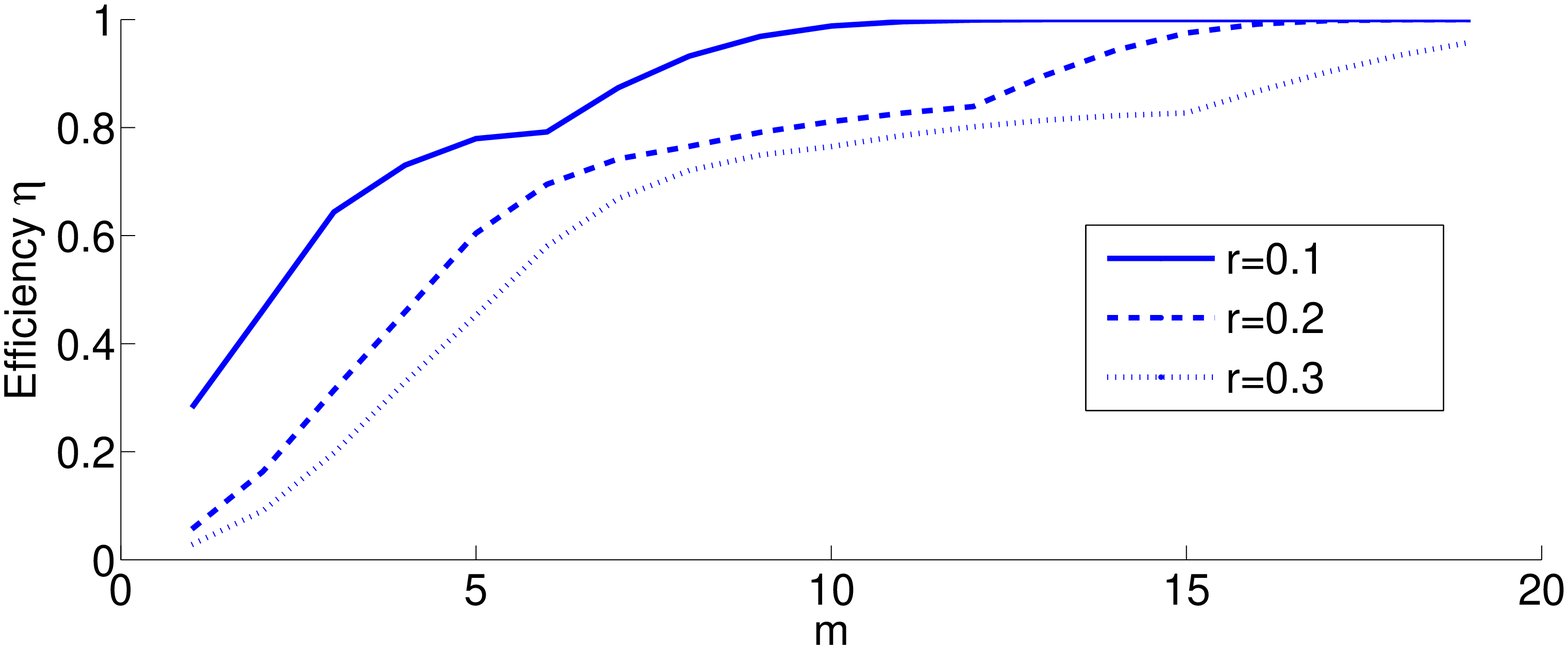}
\caption{Efficiency versus m for three different sets of values of probabilities for $l=20$ and $n=3$, $q_1=q_2=q_3=r$, $v=100, c=1$, $g_i(x)=x^{10}-i^7.$}
\label{fig:Efficiency}
\end{figure}

\section{Discussion And Future Work}
For analytical tractability, we have assumed symmetric setting. An important problem for future research is to characterize NE under asymmetric setting. We have assumed secondaries have the same penalty function. Another interesting extension would be the characterization of NE strategy profile  when the assumption is relaxed.

\section{Appendix }\label{sec:appen}
First, we introduce some terminologies and observation that we will use throughout this section.
\begin{defn}\label{dm}
 Let $X_{m}$ be the $m$th smallest offered penalty offered by  primaries $i=2,\ldots,l$, and  let $F(\cdot)$ denote the distribution function of $X_m$.
\end{defn}
For a symmetric strategy profile, $F(\cdot)$ would remain the same if we had considered any $l-1$ primaries
rather than $2, \ldots, l$.

\begin{obs}\label{o1}
Any point $y \leq g_j(c)$ can not be a best response (definition \ref{br}) for channel state $j$.
\end{obs}
The observation is evident as  the profit $\phi_j(\cdot)$  of a primary is non-positive
if the selected penalty is  $\leq g_j(c)$. But, $\phi_j(x) > 0$ for  $g_j(c)<x\leq v$ as $0<\sum_{i=1}^{n}q_i<1$.
\subsection{Proof of Section ~\ref{sec:structure}}
\textit{Proof of Theorem~\ref{thm1}}:
Suppose, $\psi_{j}(\cdot)$ has a jump\footnote{A d.f. $G(x)$ is said to have a jump at $y$ of size $a$, if $G(y)-G(y-)=b$, where $G(y-)=\lim_{x\uparrow y}G(x)$} at $x$, then all primaries select $x$ as their penalties with positive probability whenever their channel states are $j.$ As, no primary selects a penalty other than a best response with positive probability, thus, $x$ has to be a best response to primaries whenever their channel state is $j$.  Hence, Observation ~\ref{o1} entails that $x>g_j(c)$.  We will complete the proof by showing that primary 1 can attain better expected profit by choosing penalty just below $x$, when channel state is $j$.\\ We have,
\begin{equation}\label{pt1}
r(x)=P(X_{m}>x)+r(x|X_{m}=x)P(X_{m}=x)
\end{equation}
For every $\epsilon>0$, it is worthy to note the following relation
\begin{align}\label{pt2}
r(x-\epsilon)\geq P(X_{m}>x)+P(X_{m}=x)
\end{align}
Thus, from (\ref{pt1}) and (\ref{pt2}), for every $\epsilon>0$,
\begin{align}\label{ep}
r(x-\epsilon)- r(x)& \geq P(X_{m}=x)\cdot(1- r(x|X_{m}=x))\nonumber\\
&  =\gamma
\end{align}
where, $\gamma=P(X_m=x)(1- r(x|X_{m}=x))$. Note that $ r(x|X_{m}=x)<1$ and $P(X_m=x)>0$, due to symmetry and $l>m$. So, $\gamma>0$.\\
Now, expected profit to primary 1 for channel state $j$ at penalty $x-\epsilon$ is:
\begin{align}\label{n3}
 \phi_j(x-\epsilon)& =(f_j(x-\epsilon)-c)r(x-\epsilon)\quad (\text{from}(\ref{ex1}))\nonumber\\
& \geq(f_j(x-\epsilon)-c)(r(x)+\gamma)
\end{align}
Let, $\delta=f_j(x)-f_j(x-\epsilon)$, so from (\ref{n3})
\begin{align}
\phi_j(x-\epsilon)& \geq (f_j(x)-\delta-c)(r(x)+\gamma)\nonumber\\
                                     & =(f_j(x)-c)(r(x)+\gamma)-\delta(r(x)+\gamma)
\end{align} 
$f_j(\cdot)$ is continuous and strictly increasing ($f_j(x)>c$),$\gamma>0$ and independent of $\epsilon$, so $\exists \epsilon>0$,  such that 
\begin{align}
\phi_j(x-\epsilon)>(f_j(x)-c)r(x)=\phi_j(x)\nonumber
\end{align}which contradicts that $x$ is a best response and hence, the result follows.\qed

Now, we will show the following lemma and observation, which will facilitate our later analysis.

\begin{lem}\label{fcont}
$F(\cdot)$ is continuous in $[c_{min},v]$ and if $\sum_{j=1}^{n}\psi_j(y)>\sum_{j=1}^{n}\psi_j(x)$, then $F(y)>F(x)$
\end{lem}
where, $c_{min}=\underset{i\in\{1,..,n\}}{\text{min}}g_i(c)$
\begin{proof}
Suppose $a\in [c_{min},v]$.  At any time slot, the event that primary 1 selects penalty less than or equal to $a$ and state of a channel is $i\geq 1$, occurs with probability $q_i\cdot\psi_i(a)$. Hence, the event that primary 1 offers penalty less than or equal to $a$ occurs with probability $\sum_{i=1}^{n}q_i\cdot\psi_i(a)$. Thus, 
\begin{align}
F(a)=P(X_m\leq a)=w(\sum_{i=1}^{n}\psi_i(a))\quad(\text{Recall}(\ref{d4}))\nonumber
\end{align}
Continuity of $F(\cdot)$ follows from the fact that $\psi_i(\cdot), i=1,\ldots,n$ are continuous (Theorem~\ref{thm1}).

Now $l<m$, and $\sum_{i=1}^{n}q_i<1$, thus $F(\cdot)$increases if $\sum_{i=1}^{n}\psi_i(.),i\in \{1,..,n\}$ increases (as $\psi_i(.)$ is d.f. so it is non-decreasing). 
\end{proof}
Lemma~\ref{fcont} implies that $P(X_m=x)=0$ for each $x$ and thus $r(x) = 1-F(x)$. Hence, 
\begin{align}\label{ex2}
\phi_j(x)=(f_j(x)-c)(1-F(x))
\end{align}
\begin{obs}\label{sup}
Every element in the support set of $\psi_i(\cdot)$ is a best response; thus, so are $L_i, U_i$.
\end{obs}
\begin{proof}
Suppose there exists a point $z$ in the support set of $\psi_i(\cdot)$, which is not a best response. Therefore, primary 	1 plays at $z$ with probability $0$ when channel state is $i$.\\Now,  one of the following two cases must arise.\\
\textit{Case I}: $\exists$ a neighborhood \cite{rudin} of radius $\delta> 0$ around $z$, such that no point in this neighborhood is a best response. Neighborhood of radius $\delta>0$ of $z$ is an open set (theorem 2.19 of \cite{rudin}). Hence, we can eliminate that neighborhood and can attain a smaller closed set, such that its complement has probability zero under $\psi_i(\cdot)$, which is against the definition of support set. \\
\textit{Case II}:  For every $\epsilon>0$, $\exists y\in (z-\epsilon,z+\epsilon)$, such that $y$ is a best response. Then, we must have a sequence $z_k, k=1,2,\ldots$  such that each $z_k$ is a best response, and $\underset{k\rightarrow \infty}{\text{lim}}z_k=z$ \cite{rudin}. But profit to primary 1 for channel state $i$  at each of $z_k$ is $(f_i(z_k)-c)(1-F(z_k))$. Now, from continuity of $f_i(\cdot)$ and $F(\cdot)$ (lemma ~\ref{fcont})-
\begin{align}\label{seq}
 \underset{k\rightarrow \infty}{\text{lim}}\phi_i(z_k)& =(f_i(z_k)-c)(1-F(z_k))\nonumber\\
        & =(f_i(z)-c)(1-F(z))=\phi_i(z)
\end{align}
As each of $z_k, k=1,2,\ldots $ is a best response, $u_{i,max}=\phi_i(z_k), k=1,2,\ldots$. Hence, from (\ref{seq}) $u_{i,max}=\phi_i(z)$ and $z$ is a best response. We can conclude the result by noting that $U_i, L_i$ (Definition \ref{dlu}) are in the support set of $\psi_i(\cdot)$.
\end{proof}

\textit{Proof of Theorem~\ref{thm2}}:
 From Observation~\ref{sup} it is sufficient to   show that for any $x,y$ such that $c_{min}< x<y\leq v$, if $x$ is a best response when the state of the channel is $j$, then $y$ can not be a best response when the state of the channel is $i$ for $i>j$. If not consider $y> x$ such that $x, y$ are the best responses when channel states are respectively  $j, i$. Now, from Observation~\ref{o1} $f_i(y)>c, f_j(x)>c$. Also,
\begin{eqnarray}
u_{i,max} & = & (f_i(y)-c)(1-F(y)) \label{ptt1} \\
\phi_j(y)& = & (f_j(y)-c)(1-F(y))\nonumber\\
        & = & u_{i,max}.\dfrac{f_j(y)-c}{f_i(y)-c} \quad (\text{from} (\ref{ptt1}))\nonumber\\
u_{j,max}& \geq & u_{i,max}.\dfrac{f_j(y)-c}{f_i(y)-c} \label{n41}
\end{eqnarray}
Next,
\begin{eqnarray}
u_{j,max} & = & (f_j(x)-c)(1-F(x))\nonumber\\
\phi_i(x)& = &(f_i(x)-c)(1-F(x))\nonumber\\
             & = & u_{j,max}.\dfrac{f_i(x)-c}{f_j(x)-c}\label{n42}
\end{eqnarray}
Using (\ref{n41}) and (\ref{n42}), we obtain-
\begin{eqnarray}\label{n43}
\phi_i(x)\geq u_{i,max}.\dfrac{(f_j(y)-c)(f_i(x)-c)}{(f_i(y)-c)(f_j(x)-c)}
\end{eqnarray}
But, then, since $y>x$, $i>j$,  (\ref{con1}) implies that $\phi_i(x)>u_{i,max}$
which contradicts the definitions of $u_{i,max}$ and $\phi_i(x)$.
\qed

\textit{Proof Of Theorem 3}:
Suppose the statement is not true. But, it follows from Theorem~\ref{thm2} that there exists an interval $(x,y) \subseteq [L_n, v]$,
such that no primary offers penalty in the interval $(x,y)$ with positive probability. So, we must have $\tilde{a}$ such that
\begin{eqnarray}
\tilde{a}=\inf\{b\leq x:\psi_{j}(b)=\psi_{j}(x), \forall j\} \nonumber
\end{eqnarray}

By definition of $\tilde{a}$, $\tilde{a}$ is a best response  for at least one state $i$. But, as primaries do not offer penalty in the range $(\tilde{a},y)$, so from \eqref{ex2}, $\phi_i(z) > \phi_i(\tilde{a})$ for
each $z \in (\tilde{a}, y)$. This is because $F(y)=F(\tilde{a})$ and $ f_i(\tilde{a}) < f_i(z)$.  Thus, $\tilde{a}$ can not be  a best response for state i.\qed

\subsection{Proofs of Section ~\ref{sec:computation}}
\textit{proof of Lemma \ref{lu}}:
We first prove \eqref{n51} using induction, (\ref{n52}) follows from (\ref{n51}).

From theorem~\ref{thm3}, $\psi_i(\cdot)$'s support set is $[L_i,L_{i-1}]$ for $i=2,...,n$ and  $[L_1,v]$ for $i=1.$  Thus, $v$ is a best response for channel state  1 (by corollary~\ref{sup}), hence
\begin{align}
u_{1,max}& =(f_1(v)-c)(1-w(\sum_{i=1}^{n}q_i))= p_1-c
\end{align}
Thus,  \eqref{n51} holds for $i=1$ with $L_0=v$.
 Let, \eqref{n51} be true for $i=t<n$. We have to show that \eqref{n51} is satisfied for $i=t+1$ assuming that it is true for $i=t$. Thus,by induction hypothesis,
\begin{align}\label{52a}
& u_{t,max}=p_t-c= (f_t(L_{t-1})-c)(1-w_{t})
\end{align}
Now, $L_t$ is a best response for state $t$, and thus,
\begin{align}\label{n53}
\phi_t(L_t) =  (f_t(L_t)-c)(1-w_{t+1})=p_t-c
\end{align}
Now, as $L_t $ is also a best response for state $t+1$ by theorem~\ref{thm2} and~\ref{thm3}, thus
\begin{eqnarray}
\phi_{t+1}(L_t) = (f_{t+1}(L_t)-c)(1-w_{t+1})=u_{t+1,max}\nonumber
\end{eqnarray}
Thus, $u_{t+1,max}=p_{t+1}-c$ and it satisfies (\ref{n51}). Thus, \eqref{n51} follows from mathematical induction.

\eqref{n52} follows since
$(f_i(L_i)-c)(1-w_{i+1})=p_i-c$
and $g_i(\cdot)$ is the inverse of $f_i(\cdot)$.\qed

\textit{proof of Lemma~\ref{lm:computation}}:
$L_i, L_{i-1}$ are the end-points of the support set of $\psi_i(\cdot)$ from definition \ref{dlu}, and their
values have been computed in  lemma~\ref{lu}.
 We should have for $x<L_i$, $\psi_i(x)=0$ and for $x>L_{i-1}$,$\psi_i(x)=1$.  From theorem~\ref{thm3},  every point $x \in [L_i,L_{i-1}]$ is a best response for state $i$, and hence,
\begin{align}
(f_i(x)-c)(1-w(\sum_{j=i+1}^{n}q_j+q_i.\psi_i(x)))=u_{i,max}=p_i-c. \nonumber
\end{align}
Thus, the expression for $\psi_i(\cdot)$  follows. We conclude the proof by noting that  the domain and range of $w(.)$ is $[0,1]$, and $\dfrac{p_i-c}{f_i(x)-c}<1$ for $x\in [L_i,L_{i-1}]$: so $w^{-1}(.)$ is defined at $1-\dfrac{p_i-c}{f_i(x)-c}$. \qed.

\textit{proof of lemma~\ref{dcont}}:
Note that
\begin{align}
& \psi_i(L_i)=\dfrac{1}{q_i}(w^{-1}(1-\dfrac{p_i-c}{f_i(L_i)-c})-\sum_{j=i+1}^{n}q_j)\nonumber\\
           & =\dfrac{1}{q_i}(w^{-1}(w_{i+1})-\sum_{j=i+1}^{n}q_j)\quad \text{from}(\ref{n52})\nonumber\\
           & =0 \quad (\text{by} (\ref{d5}))
\end{align}
From (\ref{c5}) and (\ref{n51}), we obtain
\begin{align}\label{c7}
& \psi_i(L_{i-1})=\dfrac{1}{q_i}(w^{-1}(1-\dfrac{p_i-c}{f_i(L_{i-1})-c})-\sum_{j=i+1}^{n}q_j)\nonumber\\
& =\dfrac{1}{q_i}(w^{-1}(w_i)-\sum_{j=i+1}^{n}q_j)\nonumber\\
& =\dfrac{1}{q_i}.q_i=1\quad (\text{as} w_i=w(\sum_{j=i}^{n}q_j))
\end{align}
$w(.)$ is continuous, strictly increasing on compact set $[0, \sum_{j=1}^{n}q_j]$, so $w^{-1}$ is also continuous (theorem 4.17 in \cite{rudin}). Also, $\dfrac{p_i-c}{f_i(x)-c}$ is continuous for $x\geq L_i$ as $f_i(x)>c$, so $\psi_i(.)$ is continuous as it is a composition of two continuous functions. Again, $w^{-1}(.)$ is strictly increasing (as $w(\cdot)$ is strictly increasing), $1-\dfrac{p_i-c}{f_i(x)-c}$ is strictly increasing (as $f_i(\cdot)$ is strictly increasing), so $\psi_i(.)$ is strictly increasing on $[L_i,L_{i-1}]$ ( as it is a composition of two strictly increasing functions (theorem 4.7 in \cite{rudin}))\qed.

\subsection{Proof of Section ~\ref{sec:existence}}
First we state and prove a result (observation~\ref{recurse}). Subsequently we prove Theorem~\ref{thm4}.
\begin{obs}\label{recurse}
For $t>s, t,s\in\{1,\ldots,n\}$
\begin{align}\label{s32}
p_t-c=(p_s-c)\prod\limits_{i=s}^{t-1}\dfrac{f_{i+1}(L_i)-c}{f_i(L_i)-c}
\end{align}
\end{obs}
\begin{proof}
Since $f_i^{-1}(\cdot)=g_i$, thus from (\ref{n52}) we obtain for $i-1$
\begin{align}\label{eq:multi}
p_{i-1}-c=(f_{i-1}(L_{i-1})-c)(1-w_{i})
\end{align}
Hence, from (\ref{n51}), (\ref{imp}), and (\ref{eq:multi})
\begin{align}\label{n52b}
p_{i}-c=(p_{i-1}-c)\dfrac{f_{i}(L_{i-1})-c}{f_{i-1}(L_{i-1})-c} 
\end{align}
We obtain the result using recursion.\end{proof}

\textit{proof of Theorem \ref{thm4}}:
Fix a state $j\in \{1,\ldots,n\}$. First, we show that if a primary follows its strategy profile then it would attain a payoff of $p_j-c$ at channel state $j$. Next, we will show that if a primary unilaterally deviates from its strategy profile, then it would obtain a payoff of at most of $p_j-c$ (Case i and Case ii) when the channel state is $j$.

If state of channel of primary 1 is $i\geq 1$ and it select penalty $x$, then its expected profit is-
\begin{align}\label{e1}
\phi_i(x)& =(f_i(x)-c)r(x)\nonumber\\
& =(f_i(x)-c)(1-w(\sum_{k=1}^{n}q_k.\psi_k(x)))
\end{align}
First, suppose $x\in [L_j,L_{j-1}]$. From (\ref{e1}) and (\ref{c5}), we obtain
\begin{align}\label{s43}
\phi_j(x)&= (f_j(x)-c)(1-w(\sum_{i=1}^{n}q_i\psi_i(x))) \notag\\& =(f_j(x)-c)(1-w(\sum_{k=j+1}^{n}q_k+q_j\psi_j(x)))\nonumber\\
 & =(f_j(x)-c)(1-w(w^{-1}(1-\dfrac{p_j-c}{f_j(x)-c})))\notag\\ 
 & =p_j-c
\end{align}
Since $\psi_i(L_n)=0$ $\forall i$, we have
\begin{align}\label{paylow}
\phi_j(L_n)& =(f_j(L_n)-c)(1-w(0))=f_j(L_n)-c
\end{align}
From (\ref{paylow}) expected payoff to a primary at state $j$ at $L_n$ is $f_j(L_n)-c$. At any $y<L_n$ expected payoff to a primary at state $j$ will be strictly less than $f_j(L_n)-c$. Hence, it suffices to show that for $x\in [L_k,L_{k-1}], k\neq j, k\in\{1,..,n\}$, profit to primary 1 is at most $p_j-c$, when the channel state is $j$. 

Now, let $x\in [L_k,L_{k-1}]$. From (\ref{e1}) expected payoff at $x$
\begin{align}
\phi_j(x)& =(f_j(x)-c)(1-w(\sum_{i=1}^{n}q_i\psi_i(x)))\nonumber\\& =(f_j(x)-c)(1-w(\sum_{i=k+1}^{n}q_i+q_k\psi_k(x)))\nonumber\\
& =(f_j(x)-c)(1-w(w^{-1}(1-\dfrac{p_k-c}{f_k(x)-c})))\notag\\ 
& =\dfrac{(p_k-c)(f_j(x)-c)}{f_k(x)-c}\nonumber
\end{align}
Hence,
\begin{eqnarray}\label{ss2}
\phi_j(x)-(p_j-c)=\dfrac{(p_k-c)(f_j(x)-c)}{f_k(x)-c}-(p_j-c)
\end{eqnarray}
We will show that $\phi_j(x)-(p_j-c)$ is non-positive. As, $k\neq j$, so only the following two cases are possible.\\
\textit{Case i}:  $k<j$\\
From (\ref{con1}), (\ref{imp}) and for $i<j$, we have-
\begin{equation}\label{ss4}
\dfrac{f_i(L_{i-1})-c}{f_i(L_i)-c}>\dfrac{f_j(L_{i-1})-c}{f_j(L_i)-c} (\text{as }  L_i<L_{i-1})
\end{equation}
From observation \ref{recurse} we obtain-
\begin{align}
p_j-c& = \dfrac{(p_k-c)(f_j(L_{j-1})-c)}{f_k(L_k)-c}\prod\limits_{i=k+1}^{j-1}\dfrac{f_i(L_{i-1})-c}{f_i(L_i)-c}\nonumber\\
\end{align}
Using (\ref{ss4}) the above expression becomes
\begin{align}
p_j-c& \geq \dfrac{(p_k-c)(f_j(L_{j-1})-c)}{f_k(L_k)-c}\prod\limits_{i=k+1}^{j-1}\dfrac{f_j(L_{i-1})-c}{f_j(L_i)-c}\nonumber\\
    & =\dfrac{(p_k-c)(f_j(L_{j-1})-c)}{f_k(L_k)-c}.\dfrac{(f_j(L_k)-c)}{f_j(L_{j-1})-c}\nonumber\\
     & =\dfrac{(p_k-c)(f_j(L_k)-c)}{f_k(L_k)-c}\nonumber
\end{align}
Hence, from (\ref{ss2}), we obtain-
\begin{align}\label{ss6}
& \phi_j(x)-(p_j-c)\notag\\& \leq (p_k-c)(\dfrac{f_j(x)-c}{f_k(x)-c}-\dfrac{f_j(L_k)-c}{f_k(L_k)-c})
\end{align}
Since $x\in [L_k,L_{k-1}]$, $j>k$ and $f_k(L_k)>c$ (by (\ref{imp}); hence, from (\ref{ss6}) and  assumption 1, we have-
\begin{eqnarray}\label{s5a}
\phi_j(x)\leq p_j-c
\end{eqnarray}
\textit{Case ii}: $j<k$\\
If $f_j(x)\leq c$ then a primary gets a non-positive payoff at channel state $j$, which is strictly below $p_j-c$. Hence we consider the case when $f_j(x)>c$. Since $x\leq L_{k-1}$ thus $f_j(L_{k-1})>c$.

Now, if $i>j$ and $f_j(L_i)>c$, we have from (\ref{con1}) and (\ref{imp})-
\begin{eqnarray}\label{s46}
\dfrac{f_{i}(L_{i-1})-c}{f_j(L_{i-1})-c}<\dfrac{f_{i}(L_{i})-c}{f_j(L_{i})-c}\quad (\text{as } L_i<L_{i-1})
\end{eqnarray}
Since $f_j(L_{k-1})>c$, thus 
\begin{align}\label{eq:greaterthan0}
f_j(L_i)>c\quad(\text{for } j\leq i<k, \text{as } L_i\geq L_{k-1})
\end{align}
Now, from observation~\ref{recurse} we obtain-
\begin{equation}
\begin{aligned}
p_k-c& =(p_j-c)\prod\limits_{i=j}^{k-1}\dfrac{f_{i+1}(L_i)-c}{f_i(L_i)-c}\nonumber\\
       & =(p_j-c).\dfrac{f_k(L_{k-1})-c}{f_j(L_j)-c}\prod\limits_{i=j+1}^{k-1}\dfrac{f_{i}(L_{i-1})-c}{f_{i}(L_{i})-c}\nonumber\\
     & \leq (p_j-c).\dfrac{f_k(L_{k-1})-c}{f_j(L_j)-c}\prod\limits_{i=j+1}^{k-1}\dfrac{f_j(L_{i-1})-c}{f_j(L_{i})-c}\nonumber\\
     & (\text{from } (\ref{s46}), \& (\ref{eq:greaterthan0}))\\
     & =(p_j-c).\dfrac{f_k(L_{k-1})-c}{f_j(L_j)-c}\dfrac{f_j(L_j)-c}{f_j(L_{k-1})-c}\nonumber\\
     & =(p_j-c).\dfrac{f_k(L_{k-1})-c}{f_j(L_{k-1})-c}
\end{aligned}
\end{equation}	
Thus, from (\ref{ss2}), we obtain-
\begin{align}\label{s5}
& \phi_j(x)-(p_j-c)\notag\\& \leq (p_j-c)(\dfrac{f_k(L_{k-1})-c}{f_j(L_{k-1})-c}.\dfrac{f_j(x)-c}{f_k(x)-c}-1)\nonumber\\
   & \leq 0 (\text{as } x\leq L_{k-1}, j<k \quad\text{and from Assumption 1})
\end{align}
Hence, from (\ref{s5}), (\ref{s5a}), and (\ref{s43}), every $x\in [L_j,L_{j-1}]$ is a best response to primary 1 when channel state is $j$. Since $j$ is arbitrary, it is true for any $j\in\{1,\ldots,n\}$ and thus (\ref{c5}) constitute a Nash Equilibrium strategy profile.\qed

\subsection{Proofs of Section ~\ref{numerical}}
We will first establish part 1 and 3 of lemma \ref{eff} . Part 2 of lemma \ref{eff} is cumbersome and we defer its proof until the end of the section. Lemma \ref{thresh} will readily follow from part 1 and part 3 of lemma \ref{eff}.
\textit{proof of part 1 of Lemma \ref{eff}}:We first present the essence of the proof. \\
Since a primary can attain at most a payoff of $f_i(v)-c$\footnote{it is attained when primary selects penalty $v$ and its channel is bought with probability 1} at channel state $i$. Thus, we have an upper bound of $R_{NE}$. Since $1-w_i\rightarrow1, i=1,\ldots,n$ a primary also attains at least a payoff of $f_i(v)-c$ at channel state $i$ in the asymptotic limit. Detailed argument follows:

Note that a primary  can achieve profit of at most $f_i(v)-c$, when channel state is $i\geq 1$. Hence,
\begin{align}\label{upn}
R_{NE}\leq \sum_{i=1}^{n}q_i\cdot(f_i(v)-c)
\end{align}
When primary 1 selects penalty $v$ at channel state $i\geq 1$, then its expected profit is $\phi_i(v)=(f_i(v)-c)(1-w_1)$. Now,
from theorem \ref{thm4} under the NE strategy profile,
\begin{align}
p_i-c\geq \phi_i(v)=(f_i(v)-c)(1-w_1)
\end{align}
Hence,
\begin{eqnarray}\label{se5}
R_{NE}\geq l.(\sum_{i=1}^{n}q_i.(f_i(v)-c))(1-w_1)
\end{eqnarray}
Let $Z_i, i=1,..,l-1$ be Bernoulli trials with success probabilities $\sum_{j=1}^{n}q_i$ and $Z=\sum_{i=1}^{l-1}Z_i$; so $P(Z\geq m)$ is equal to $w_1$ by (\ref{d4}) and (\ref{n50}).  Since $m\geq (l-1)(\sum_{i=1}^{n}q_i+\epsilon)$ and $E(Z)=(l-1)\sum_{i=1}^{n}q_i$, by weak law of large numbers \cite{ross}, $w_1\rightarrow 0$ as $l\rightarrow \infty$. Hence the result follows from (\ref{upn}) and (\ref{se5}).\qed.

\textit{proof of part 3 of Lemma \ref{eff}}: We first provide an outline of the proof.

When $m\leq (l-1)(q_n-\epsilon)$ for some $\epsilon>0$, an application of Hoeffding\rq{}s inequality shows that $1-w_n$ approaches $0$ as $l\rightarrow\infty$. Since $1-w_n>1-w_j$  for $j<n$, thus $1-w_j$ approaches $0$ as $l->\infty$. We subsequently obtain upper bounds of $R_{NE}$ in terms of $1-w_j , j=1,2\ldots,n$ which in turn proves the desired result. Detailed argument follows:

Suppose that $m\leq (l-1)(q_n-\epsilon)$, for some $\epsilon>0$. Let, $Z_i, i=1,...,l-1$ be the Bernoulli trials with success probabilities $q_n$ and $Z=\sum_{i=1}^{l-1}Z_i$, $E(Z)=(l-1)q_n$. Hence,
\begin{eqnarray}\label{e12}
1-w_n& \leq&  P(Z\leq m)\nonumber\\
& & \leq P(Z\leq (l-1)(q_n-\epsilon))\nonumber\\
& &   \leq P(|Z-(l-1)q_n|\geq (l-1)\epsilon)\nonumber\\
& & \leq 2\exp(-\dfrac{2(l-1)^2\epsilon^2}{l-1})\nonumber\\
& & (\text{from Hoeffding's Inequality \cite{hoeffding}})\nonumber\\
& & =2\exp(-2(l-1)\epsilon^2)
\end{eqnarray}
Note that 
$1-w_i<1-w_j$ (if $j>i$), $f_k(L_{k-1})>f_{k-1}(L_{k-1})$. Hence, it can be readily seen from(\ref{n51}) that 
\begin{align}
& p_i-c\leq (f_i(L_{i-1})-c)(1-w_n)
\end{align}
Thus,
\begin{align}\label{nz}
R_{NE}& l.\leq (1-w_n)(\sum_{j=1}^{n}q_j.(f_j(L_{j-1})-c))
\end{align}
As $f_i(c)\leq L_i<L_{i-1}\leq v$, hence, the result follows from (\ref{e12}) \qed.

Note that the bound of $R_{NE}$ (from (\ref{e12}) and (\ref{nz})) for $m\leq (l-1)(q_n-\epsilon)$, $\epsilon>0$,
\begin{equation}\label{exnz}
R_{NE}\leq l.\gamma\cdot\exp(-2\epsilon^2.(l-1))
\end{equation}
where $\gamma=2(1-w_n)(\sum_{j=1}^{n}q_j.(f_j(L_{j-1})-c))$. We will use this bound in proving the part 2 of lemma~\ref{thresh}.

From, the definition of $\eta$, it should be clear that
\begin{equation}\label{bound}
\eta\leq 1
\end{equation}

Now, we show lemma \ref{thresh}

\textit{proof of part 1 of  lemma \ref{thresh}}:
First suppose that $m\geq (l-1)(\sum_{i=1}^{n}q_i+\epsilon)$.
From, definition of $R_{OPT}$, it is obvious that 
\begin{eqnarray}\label{e4}
R_{OPT}\leq l\cdot(\sum_{i=1}^{n}(q_i.(f_i(v)-c)))
\end{eqnarray}
Hence the result follows from part 1 of lemma \ref{eff}, (\ref{e4}) and (\ref{bound}).\qed

\textit{proof of part 2 of Lemma \ref{thresh}}:
Suppose that $m\leq (l-1)(q_n-\epsilon)$, for some $\epsilon>0$. We prove the result by showing that $R_{NE}$ decreases at fast rate to $0$ compared to $R_{OPT}$ when $l\rightarrow\infty$.\\
 Let, $Z$ be the number of primaries, whose channel is in state $n$. Hence,
\begin{align}\label{boundropt}
& R_{OPT}\geq E(\min(Z,m))(f_n(v)-c)\nonumber\\& \dfrac{R_{OPT}}{f_n(v)-c}\geq E(\min(Z,m))
\end{align}
Note that $E(Z)=l\cdot q_n$, $Var(Z)=l\cdot q_n(1-q_n)$.\\ We introduce a new random variable $Y$ as follows-
\begin{equation*}
Y=\begin{cases} m, &  \text{if} Z\geq m\\
0, & \text{otherwise} \end{cases}
\end{equation*}
So, 
\begin{align}\label{e10}
E(\min(Z,m))& \geq  E(Y)\nonumber\\
& =m.P(Z\geq m)\nonumber\\
& =m.(1-P(Z<m))\nonumber\\
& \geq m.(1-P(Z\leq (l-1)(q_n-\epsilon)))\nonumber\\
& \geq m.(1-P(|Z-l.q_n|\geq (l-1)\epsilon)\nonumber\\
& \geq m.(1-\dfrac{l.q_n.(1-q_n)}{(l-1)^2.\epsilon^2}) \nonumber\\& (\text{From Chebyshev's Inequality})
\end{align}
Hence, from (\ref{exnz}), (\ref{boundropt}) and (\ref{e10}), we obtain-
\begin{eqnarray}\label{e8}
\eta\leq \dfrac{l.\gamma.\exp(-2(l-1)\epsilon^2)}{m.(1-\dfrac{l.q_n.(1-q_n)}{(l-1)^2.\epsilon^2}).(f_n(v)-c)} \nonumber
\end{eqnarray}
Thus, $\eta$ tends to zero for $m\leq (l-1)(q_n-\epsilon)$,as $l$ tends to infinity (as $m\neq 0$).\qed.

\textit{proof of part 3 of Lemma \ref{eff}}:

We show the result by evaluating the expressions for $p_j-c, j=1,\ldots,n$ in the asymptotic limit. Towards this end, we first evaluate the expressions for $w_j$ and $L_j$ in the asymptotic limit. We obtain the expression for $p_j-c$ when we combine those two values.\\
Suppose that $(l-1)(\sum_{j=i-1}^{n}q_j-\epsilon)\geq m\geq (l-1)(\sum_{j=i}^{n}q_j+\epsilon), i\in\{2,\ldots,n\}$ for some $\epsilon>0$. Since $w_i$ is the probability of at least $m$ successes out of $l-1$ independent Bernoulli trials, each of which occurs with probability $\sum_{j=i}^{n}q_j$ (by (\ref{d5})). Hence from the weak law of large numbers \cite{ross}
\begin{align}\label{asy1}
& w_i\rightarrow 0 \quad \text{as } l\rightarrow \infty\nonumber\\
& 1-w_i\rightarrow 1\quad \text{as } l\rightarrow \infty
\end{align}
Since $w_j<w_i$, for $j>i$ (from (\ref{d5})), we have from (\ref{asy1}) for $j\geq i$
\begin{equation}\label{asy2}
1-w_j\rightarrow 1\quad \text{as } l\rightarrow \infty
\end{equation}

Again, as $m\leq(l-1)(\sum_{j=i-1}^{n}q_j-\epsilon)$, so,  from weak law of large numbers\cite{ross},  for every $\epsilon>0$, $\exists L$, such that $1-w_{i-1}<\epsilon$, whenever $l\geq L$. Hence,
\begin{align}\label{asyz1}
1-w_{i-1}& \underset{l\rightarrow \infty}{\rightarrow} 0\nonumber\\
1-w_{j}& \underset{l\rightarrow \infty}{\rightarrow} 0 \quad (\text{for }j< i)
\end{align}
Thus, it is evident from (\ref{n51}) and (\ref{asyz1}) that
\begin{align}\label{asypaylo}
p_j-c\underset{l\rightarrow \infty}{\rightarrow} 0 \quad (\text{for }j<i) 
\end{align}
Thus, from (\ref{n52}), (\ref{asy2}), and (\ref{asypaylo})
\begin{align}\label{asyl}
L_{i-1}\underset{l\rightarrow \infty}{\rightarrow} c_{i-1}
\end{align}
We obtain for $j\geq i$ from (\ref{n51}) 
\begin{align}\label{asyimp}
p_j-c& =(f_j(L_{j-1})-c)(1-w_j)\nonumber\\
p_j& \underset{l\rightarrow \infty}{\rightarrow} f_j(L_{j-1})\quad (\text{from }(\ref{asy2}))
\end{align}
Again, using (\ref{n52}), we obtain for $j\geq i$
\begin{align}\label{asyimp2}
p_j-c& =(f_j(L_j)-c)(1-w_{j+1})\nonumber\\
p_j & \underset{l\rightarrow \infty}{\rightarrow} f_j(L_{j})\quad (\text{from }(\ref{asy2}))
\end{align}
$f_j(\cdot)$ is strictly increasing, thus from (\ref{asyimp}) and (\ref{asyimp2}), $L_j\rightarrow L_{j-1}$ (for $j\geq i$). Hence, for $j\geq i$,
\begin{align}\label{asylj}
L_j & \underset{l\rightarrow \infty}\rightarrow L_{i-1}\nonumber\\
L_j & \underset{l\rightarrow \infty}{\rightarrow} c_{i-1}\quad (\text{from}(\ref{asyl}))
\end{align}                               
Thus, from (\ref{asylj}), and  (\ref{asyimp2}), we obtain for $j\geq i$ 
\begin{align}\label{asyhi}
p_j-c\underset{l\rightarrow \infty}{\rightarrow} (f_j(c_{i-i})-c)
\end{align}
Thus, from (\ref{asypaylo}) and (\ref{asyhi}), we have-
\begin{eqnarray}
R_{NE}\underset{l\rightarrow \infty}{\rightarrow} l.\sum_{j=i}^{n}q_j\cdot(f_j(c_{i-i})-c)\nonumber
\end{eqnarray}
which is equal to the required expression.\qed

\end{document}